\documentclass[a4paper,UKenglish]{lipics-v2018}
\nolinenumbers
\usepackage[utf8]{inputenc}
\usepackage{tikz}
\usepackage{verbatim}
\usetikzlibrary{arrows}
\usepackage{caption}
\usepackage{float}
\usepackage{amsmath}
\usepackage{amsthm}
\usepackage{amssymb}

\newtheorem*{problem*}{Problem}
\newtheorem*{theorem*}{Theorem}
\newtheorem*{lemma*}{Lemma}
\newtheorem*{definition*}{Definition}
\newtheorem*{corollary*}{Corollary}
\newtheorem*{remark*}{Remark}
\usepackage{algorithm}
\usepackage{algpseudocode}
\newcommand\andrei[1]{\textcolor{blue}{#1}}

\title{A connection between String Covers and Cover Deterministic Finite Tree Automata Minimization}
\titlerunning{String Covers and Cover Deterministic Finite Tree Automata}

\author{Alexandru Popa}{University of Bucharest, Bucharest, Romania \\ National Insitute of Research and Development in Informatics, Bucharest, Romania}{alexandru.popa@fmi.unibuc.ro}{}{}

\author{Andrei T\u{a}n\u{a}sescu}{Politehnica University of Bucharest, Bucharest, Romania}{andrei.tanasescu@mail.ru}{}{}

\authorrunning{A., Popa and A., T\u an\u asescu}

\Copyright{Alexandru Popa and Andrei T\u an\u asescu}

\subjclass{Graph algorithms analysis}

\keywords{string cover, periodicity, polynomial time exact algorithms, automata}

\category{}

\relatedversion{}

\supplement{}

\funding{}

\EventEditors{}
\EventNoEds{0}
\EventLongTitle{}
\EventShortTitle{}
\EventAcronym{}
\EventYear{}
\EventDate{}
\EventLocation{}
\EventLogo{}
\SeriesVolume{}
\ArticleNo{}

\begin{document}

\maketitle{}

\begin{abstract}
Data compression plays a crucial part in the  cloud based systems of today. One the fundaments of compression is quasi-periodicity, for which there are several models. We build upon the most popular quasi-periodicity model for strings, i.e., covers, generalizing it to trees. 

We introduce a new type of cover automata, which we call \textbf{D}eterministic \textbf{T}ree \textbf{A}utomata. Then, we formulate a cover problem on these DTA and study its complexity, in both sequential and parallel settings. We obtain bounds for the Cover Minimization Problem. Along the way, we uncover an interesting application, the Shortest Common Cover Problem, for which we give an optimal solution.
\end{abstract}

\section{Introduction}

Redundancy is an important phenomenon in engineering and computer science ~\cite{ming1990kolmogorov,muchnik2003almost}. One of the most important aspects of redundancy is periodicity. In turn, periodicity is a very important phenomenon when analyzing physical data such as an analogue signal. In general, natural data is very redundant or repetitive and exhibits some key patterns or regularities~\cite{HAVLIN1995171,timmermans2017cyclical,tychonoff1935theoremesL} which we may assert to be the actual $intention$~\cite{searle1980speech} of the data. Periodicity itself has been thoroughly studied in various fields such as Signal Processing~\cite{sethares1999periodicity}, Bioinformatics~\cite{brodzik2007quaternionic}, Dynamical Systems~\cite{katok1997introduction} and Control Theory~\cite{bacciotti2006liapunov}, each bringing its own insights. 

Due to the inherent imperfection of natural data, it is highly unlikely that the data is periodic. In fact, the data is \emph{almost}(\emph{quasi-}) periodic\cite{ApostolicoB97}. This has been firstly studied over strings, the most general representation of digital data~\cite{middlestead2017digital}. 

For example, assume that we  want to send the word $aba$ over a noisy channel as a digital signal where letters are modulated using \emph{amplitude shift keying}~\cite{middlestead2017digital}. Since, the simple transmission is unlikely to  yield the result due to the imperfect transmission channel, we add redundancy and thus send the word $aba$ multiple times. However, when errors occur the received signal only partially retains its periodicity. 

\begin{figure}
\noindent\begin{minipage}{\textwidth}
\begin{minipage}[c][4cm][c]{\dimexpr0.5\textwidth-5pt\relax}
\resizebox{\textwidth}{!}{
\begin{tikzpicture}
	\draw[thick, black] (-1,0) -- (0,0);
    \draw[thick, black] (12.3,0) -- (13.3,0);

    \draw[thick, black] (0,0) sin (0.25,2) cos (0.5,0) sin (0.75,-2) cos (1,0) -- (1.1,0);
    \draw[thick, black] (1.1,0) sin (1.35,1) cos (1.6,0) sin (1.85,-1) cos (2.1,0) -- (2.2,0);
    \draw[thick, black] (2.2,0) sin (2.45,2) cos (2.7,0) sin (2.95,-2) cos (3.2,0) -- (3.4,0);
    
    \draw[thick, black] (3.4,0) sin (3.65,2) cos (3.9,0) sin (4.15,-2) cos (4.4,0) -- (4.5,0);
    \draw[thick, black] (4.5,0) sin (4.75,1) cos (5.0,0) sin (5.25,-1) cos (5.5,0) -- (5.6,0);
    \draw[thick, black] (5.6,0) sin (5.85,2) cos (6.1,0) sin (6.35,-2) cos (6.6,0) -- (6.8,0);
    
    \draw[thick, black] (6.8,0) sin (7.05,2) cos (7.3,0) sin (7.55,-2) cos (7.8,0) -- (7.9,0);
    \draw[thick, black] (7.9,0) sin (8.15,1) cos (8.4,0) sin (8.65,-1) cos (8.9,0) -- (9.0,0);
    \draw[thick, black] (9.0,0) sin (9.25,2) cos (9.5,0) sin (9.75,-2) cos (10.0,0) -- (10.1,0);
    
    \draw[thick, dashed, black] (9.0,0) -- (9.1,0) sin (9.35,2) cos (9.6,0) sin (9.85,-2) cos (10.1,0) -- (10.2,0);
    \draw[thick, black] (10.2,0) sin (10.45,1) cos (10.7,0) sin (10.95,-1) cos (11.2,0) -- (11.3,0);
    \draw[thick, black] (11.3,0) sin (11.55,2) cos (11.8,0) sin (12.05,-2) cos (12.3,0);

	\draw (0.5,2) node [anchor=south] {$a$};
	\draw (1.6,2) node [anchor=south] {$b$};
	\draw (2.7,2) node [anchor=south] {$a$};
    
	\draw (3.9,2) node [anchor=south] {$a$};
	\draw (5.0,2) node [anchor=south] {$b$};
	\draw (6.1,2) node [anchor=south] {$a$};
    
	\draw (7.3,2) node [anchor=south] {$a$};
	\draw (8.4,2) node [anchor=south] {$b$};
	\draw (9.5,2) node [anchor=south] {$a$};
    
	\draw (10.7,2) node [anchor=south] {$b$};
	\draw (11.8,2) node [anchor=south] {$a$};

\end{tikzpicture}}
\end{minipage}\hfill
\begin{minipage}[c][4cm][c]{\dimexpr0.5\textwidth-5pt\relax}
\resizebox{\textwidth}{!}{
\begin{tikzpicture}
	\draw[thick, black] (-1,0) -- (0,0);
    \draw[thick, black] (12.3,0) -- (13.3,0);

    \draw[thick, black] (0,0) sin (0.25,2) cos (0.5,0) sin (0.75,-2) cos (1,0) -- (1.1,0);
    \draw[thick, black] (1.1,0) sin (1.35,1) cos (1.6,0) sin (1.85,-1) cos (2.1,0) -- (2.2,0);
    \draw[thick, black] (2.2,0) sin (2.45,2) cos (2.7,0) sin (2.95,-2) cos (3.2,0) -- (3.4,0);
    
    \draw[thick, black] (3.4,0) sin (3.65,2) cos (3.9,0) sin (4.15,-2) cos (4.4,0) -- (4.5,0);
    \draw[thick, black] (4.5,0) sin (4.75,1) cos (5.0,0) sin (5.25,-1) cos (5.5,0) -- (5.6,0);
    \draw[thick, black] (5.6,0) sin (5.85,2) cos (6.1,0) sin (6.35,-2) cos (6.6,0) -- (6.8,0);
    
    \draw[thick, black] (6.8,0) sin (7.05,2) cos (7.3,0) sin (7.55,-2) cos (7.8,0) -- (7.9,0);
    \draw[thick, black] (7.9,0) sin (8.15,1) cos (8.4,0) sin (8.65,-1) cos (8.9,0) -- (9.0,0);
    \draw[thick, black] (9.0,0) sin (9.3,3.9) cos (9.55,0) sin (9.8,-3.9) cos (10.1,0) -- (10.2,0);
    \draw[thick, dashed, blue] (9.0,0) sin (9.25,2) cos (9.5,0) sin (9.75,-2) cos (10.0,0) -- (10.2,0);
    
    \draw[thick, dashed, red] (9.0,0) -- (9.1,0) sin (9.35,2) cos (9.6,0) sin (9.85,-2) cos (10.1,0) -- (10.2,0);
    \draw[thick, black] (10.2,0) sin (10.45,1) cos (10.7,0) sin (10.95,-1) cos (11.2,0) -- (11.3,0);
    \draw[thick, black] (11.3,0) sin (11.55,2) cos (11.8,0) sin (12.05,-2) cos (12.3,0);

	\draw (0.5,2) node [anchor=south] {$a$};
	\draw (1.6,2) node [anchor=south] {$b$};
	\draw (2.7,2) node [anchor=south] {$a$};
    
	\draw (3.9,2) node [anchor=south] {$a$};
	\draw (5.0,2) node [anchor=south] {$b$};
	\draw (6.1,2) node [anchor=south] {$a$};
    
	\draw (7.3,2) node [anchor=south] {$a$};
	\draw (8.4,2) node [anchor=south] {$b$};
	\draw (9.5,4) node [anchor=south] {$c$};
    
	\draw (10.7,2) node [anchor=south] {$b$};
	\draw (11.8,2) node [anchor=south] {$a$};

\end{tikzpicture}}
\end{minipage}

\begin{minipage}[c][4em][t]{\dimexpr0.5\textwidth-5pt\relax}
\captionof{figure}{The string $aba$ sent repeatedly over a channel as an ASK signal, with a desynchronization moment}
\end{minipage}\hfill
\begin{minipage}[c][4em][t]{\dimexpr0.5\textwidth-5pt\relax}
\captionof{figure}{The string $aba$ sent repeatedly over a channel as an ASK signal, with an echo}
\end{minipage}\hfill
\end{minipage}
\end{figure}

\subsection{Stringology}

Quasi-periodicity was introduced by Ehrenfeucht in 1990 (according to ~\cite{ApostolicoB97}), in a Tech Report for Purdue University, even though it was not published in Elsevier until 1993~\cite{apostolico1993efficient}. The first paper that considered quasi-periodicity in Computer Science~\cite{ApostolicoFI91} defined the quasi-period of a string to be the length of its shortest cover and presented an algorithm for computing it, linear in both time and space. This attracted the attention of various researchers ~\cite{breslauer1992line,breslauer1994testing,li2002computing,moore1994optimal,moore1995correction} and the first decade of results can be found in several surveys ~\cite{apostolico1997periods,kociumaka2015fast,kolpakov2003finding}.

However, quasi-periodicity takes many forms, depending on the type of patterns we want to recover. Further work has been concerned with different descriptors such as seeds~\cite{guo2006computing}, the maximum quasi-periodic substring~\cite{pedersen2000finding}, k-covers~\cite{cole2005complexity}, $\lambda$-covers~\cite{guo2006computing}, enhanced covers~\cite{flouri2013enhanced}, partial covers~\cite{kociumaka2015fast}. Another variation point is the context, e.g. indeterminate strings~\cite{antoniou2008conservative} or weighted sequences~\cite{christodoulakis2006computation}. Some of the related problems are $\mathcal{NP}$-hard, e.g.~\cite{AmirLLP17}.

For some applications such as molecular biology and computer-assisted musical analysis the definition of periodicity is loosened. Thus, quasi-periodicity takes the form of \emph{approximate repetitions}. We may define an approximatively repeating pattern as a substring whose occurrences leave very few gaps, or that all repetitions are near an ``original'' source. Landau and Schmidt were the first to study this form of quasi-periodicity and focused on approximate tandem repeats~\cite{landau2001algorithm}.

\begin{definition}[String Covers]
Given a string $w$ over an alphabet, $\Sigma$, a string $s$ covers $w$ if any character of $w$ belongs to some occurrence of $s$ in $w$.
\end{definition}

Determining the shortest cover of a given string $w$ is called the Minimal \textbf{S}tring \textbf{C}over \textbf{P}roblem (\textbf{SCP} for short). Apostolico et al. prove that \textbf{SCP} is solvable in linear time \cite{ApostolicoFI91}. 

However, sometimes it might be that overlaps lead to strange behaviour such as corruptions or specific occurences getting corrupted. Nonetheless, if the number of deviations is small we may try to recover the original cover. This leads to the following definition.

\begin{definition}[Approximate String Covers]
Given a string $w$ over an alphabet, $\Sigma$, a string $s$ is the approximate string cover of $w$ if it covers $w^\prime$, the closest string to $w$, under some metric, that admits a cover.
\end{definition}

Determining the approximate cover of a given string $w$ is called the \textbf{A}pproximate String \textbf{C}over \textbf{P}roblem (\textbf{ACP} for short). Amir et al. prove that \textbf{ACP} is $\mathcal{NP}$-hard with respect to the Hamming distance~\cite{AmirLLP17}. 

Determining the string cover can be done in linear time and space while determining the approximate string cover is an $\mathcal{NP}$-hard problem. Thus, a simple variation in the definition greatly impacts its difficulty. On the other hand, the more we deviate from the original definition the richer the encoded meaning is. As we explain previously, while string covers encode desynchronization, they cannot encode corruption, for which we must use approximate string covers. Of particular interest is the information we can encode in such a cover problem while keeping it polynomial-time solvable. 

\subsection{A Formal Language Approach}
Formal languages present a natural playing field for the task at hand. The most particular class of formal languages is that of \emph{finite languages} i.e. the class $\mathcal{L}_{finite}\left(\Sigma\right)=\lbrace\mathcal{L}\subseteq\Sigma^*\vert\exists n\in\mathbb{N},\,\lvert\mathcal{L}\rvert\leq n\rbrace$. In particular, these languages are also \emph{regular}. For example, a language $\mathcal{L}_{finite}\left(\Sigma\right)\ni\mathcal{L}=\lbrace w_i\vert i\in\overline{1,\,\lvert\mathcal{L}\rvert}\rbrace$ can be described by the regular expression $E_\mathcal{L}=\left(w_1\right)\vert\left(w_2\right)\vert\dots\vert\left(w_{\lvert\mathcal{L}\rvert}\right)$. 

Yet a fundamental result is that this expression $E_\mathcal{L}$ has an associated \emph{trimmed deterministic real-time finite state automaton}, i.e. a tuple $\mathfrak{A}_\mathcal{L}=\left(Q,\,\Sigma,\,\delta,\,q_0,\,F\right)$ with $Q$ \emph{finite} a set of \emph{states},\, $\delta:Q\times \Sigma\rightarrow Q$ a partial transition function, $q_0\in Q$ such that $\forall w\in\Sigma^*,\, w\in\mathcal{L}\Leftrightarrow\delta\left(q_0,\,w\right)=\delta\left(\delta\left(\dots\delta\left(q_0,\,w^1\right)\dots,\,w^{\lvert w\rvert-1}\right),\,w^{\lvert w\rvert}\right)\in F$ discriminates (\emph{accepts} or \emph{rejects}) words, and, additionally, any state can potentially accept i.e. $\forall q\in Q\,\exists w\in\Sigma^*,\, \delta\left(q,\,w\right)\in F$. The automaton is \emph{real time} in that \emph{characters} ($w^i$) are processed one at a time, \emph{deterministic} in that a character is always processed the same way for a given state, and \emph{trimmed} because we have no useless states (effectively equivalent with the sink state $\bot$ in complete automata). The language $\mathcal{L}$ is encoded as the set of (\emph{possible infinite}) set of paths between $q_0$ and $F$.

\begin{figure}[!htb]
\begin{minipage}{0.4\textwidth}
\vspace*{2cm}

\resizebox{\textwidth}{!}{

\begin{tikzpicture}
\draw[->, thick] (-0.7,0) -- (-0.3,0);
\draw (0,0) circle (0.3);
\node at (0,0) {$q_0$};
\node[anchor=south] at (0.5,0) {$a$};
\draw[->, thick] (0.3,0) -- (0.7,0);
\draw (1,0) circle (0.3);
\node at (1,0) {$q_1$};
\node[anchor=south] at (1.5,0) {$b$};
\draw[->, thick] (1.3,0) -- (1.7,0);
\draw (2,0) circle (0.3);
\node at (2,0) {$q_2$};
\node[anchor=south] at (2.5,0) {$a$};
\draw[->, thick] (2.3,0) -- (2.7,0);
\draw (3,0) circle (0.3);
\draw (3,0) circle (0.25);
\node at (3,0) {$q_3$};
\node[anchor=south] at (2.5,0.6) {$b$};
\draw[->, thick] (3,0.3) to [out=90, in=90] (2,0.3);
\node[anchor=north] at (2,-0.6) {$a$};
\draw[->, thick] (3,-0.3) to [out=-150, in=-30] (1,-0.3);
\end{tikzpicture}}

\vspace*{3.9cm}
\caption{An automaton accepting the regular expression $E_\mathcal{L}=\left(aba\right)\left(\left(aba\right)^*\left(ba\right)^*\right)^*$ i.e. the language $\mathcal{L}=\lbrace w\in\Sigma^*\vert w \textmd{ is covered by } aba\rbrace$} 
\end{minipage}
\hspace{0.15cm}
\begin{minipage}{0.55\textwidth}
\resizebox{\textwidth}{!}{

\begin{tikzpicture}
\draw[->, thick] (-0.7,0) -- (-0.3,0);
\draw (0,0) circle (0.3);
\node at (0,0) {$q_0$};
\node[anchor=south,red] at (0.5,0) {$a$};
\draw[->, thick,red] (0.3,0) -- (0.7,0);
\draw (1,0) circle (0.3);
\node at (1,0) {$q_1$};
\node[anchor=south,red] at (1.5,0) {$b$};
\draw[->, thick,red] (1.3,0) -- (1.7,0);
\draw (2,0) circle (0.3);
\node at (2,0) {$q_2$};
\node[anchor=south,red] at (2.5,0) {$a$};
\draw[->, thick,red] (2.3,0) -- (2.7,0);
\draw (3,0) circle (0.3);
\draw (3,0) circle (0.25);
\node at (3,0) {$q_3$};
\node[anchor=south, rotate=36.83,red] at (3.5,0.3) {$a$};
\draw[->, thick,red] (3.3,0) -- (3.7,0.6);
\draw (4,0.6) circle (0.3);
\node at (4,0.6) {$q_4^1$};
\node[anchor=south, rotate=-36.83] at (3.5,-0.3) {$b$};
\draw[->, thick] (3.3,0) -- (3.7,-0.6);
\draw (4,-0.6) circle (0.3);
\node at (4,-0.6) {$q_4^2$};
\node[anchor=south,red] at (4.5,0.6) {$b$};
\draw[->, thick,red] (4.3,0.6) -- (4.7,0.6);
\draw (5,0.6) circle (0.3);
\node at (5,0.6) {$q_5^1$};
\node[anchor=south] at (4.5,-0.6) {$a$};
\draw[->, thick] (4.3,-0.6) -- (4.7,-0.6);
\draw (5,-0.6) circle (0.25);
\draw (5,-0.6) circle (0.3);
\node at (5,-0.6) {$q_5^2$};
\node[anchor=south,red] at (5.5,0.6) {$a$};
\draw[->, thick,red] (5.3,0.6) -- (5.7,0.6);
\draw (6,0.6) circle (0.3);
\draw (6,0.6) circle (0.25);
\node at (6,0.6) {$q_6^1$};
\node[anchor=south] at (5.5,-0.6) {$a$};
\draw[->, thick] (5.3,-0.6) -- (5.7,-0.6);
\draw (6,-0.6) circle (0.3);
\node at (6,-0.6) {$q_6^2$};
\node[anchor=south,rotate=-71.56] at (5.5,-1.2) {$b$};
\draw[->, thick] (5.3,-0.6) -- (5.7,-1.8);
\draw (6,-1.8) circle (0.3);
\node at (6,-1.8) {$q_6^3$};
\node[anchor=south,rotate=71.56,red] at (6.5,1.2) {$a$};
\draw[->, thick,red] (6.3,0.6) -- (6.7,1.8);
\draw (7,1.8) circle (0.3);
\node at (7,1.8) {$q_7^1$};
\node[anchor=south] at (6.5,0.6) {$b$};
\draw[->, thick] (6.3,0.6) -- (6.7,0.6);
\draw (7,0.6) circle (0.3);
\node at (7,0.6) {$q_7^2$};
\node[anchor=south] at (6.5,-0.6) {$b$};
\draw[->, thick] (6.3,-0.6) -- (6.7,-0.6);
\draw (7,-0.6) circle (0.3);
\node at (7,-0.6) {$q_7^3$};
\node[anchor=south] at (6.5,-1.8) {$a$};
\draw[->, thick] (6.3,-1.8) -- (6.7,-1.8);
\draw (7,-1.8) circle (0.3);
\draw (7,-1.8) circle (0.25);
\node at (7,-1.8) {$q_7^4$};
\node[anchor=south,red] at (7.5,1.8) {$b$};
\draw[->, thick,red] (7.3,1.8) -- (7.7,1.8);
\draw (8,1.8) circle (0.3);
\node at (8,1.8) {$q_8^1$};
\node[anchor=south] at (7.5,0.6) {$a$};
\draw[->, thick] (7.3,0.6) -- (7.7,0.6);
\draw (8,0.6) circle (0.3);
\draw (8,0.6) circle (0.25);
\node at (8,0.6) {$q_8^2$};
\node[anchor=south] at (7.5,-0.6) {$a$};
\draw[->, thick] (7.3,-0.6) -- (7.7,-0.6);
\draw (8,-0.6) circle (0.3);
\draw (8,-0.6) circle (0.25);
\node at (8,-0.6) {$q_8^3$};
\node[anchor=south] at (7.5,-1.8) {$a$};
\draw[->, thick] (7.3,-1.8) -- (7.7,-1.8);
\draw (8,-1.8) circle (0.3);
\node at (8,-1.8) {$q_8^4$};
\node[anchor=south,rotate=-71.56] at (7.5,-2.4) {$b$};
\draw[->, thick] (7.3,-1.8) -- (7.7,-3);
\draw (8,-3) circle (0.3);
\node at (8,-3) {$q_8^5$};
\node[anchor=south,rotate=71.56,red] at (8.5,2.4) {$a$};
\draw[->, thick,red] (8.3,1.8) -- (8.7,3);
\draw (9,3) circle (0.3);
\draw (9,3) circle (0.25);
\node at (9,3) {$q_9^1$};
\node[anchor=south,rotate=71.56] at (8.5,1.2) {$a$};
\draw[->, thick] (8.3,0.6) -- (8.7,1.8);
\draw (9,1.8) circle (0.3);
\node at (9,1.8) {$q_9^2$};
\node[anchor=south] at (8.5,0.6) {$b$};
\draw[->, thick] (8.3,0.6) -- (8.7,0.6);
\draw (9,0.6) circle (0.3);
\node at (9,0.6) {$q_9^3$};
\node[anchor=south] at (8.5,-0.6) {$a$};
\draw[->, thick] (8.3,-0.6) -- (8.7,-0.6);
\draw (9,-0.6) circle (0.3);
\node at (9,-0.6) {$q_9^4$};
\node[anchor=south,rotate=-71.56] at (8.5,-1.2) {$b$};
\draw[->, thick] (8.3,-0.6) -- (8.7,-1.8);
\draw (9,-1.8) circle (0.3);
\node at (9,-1.8) {$q_9^5$};
\node[anchor=south,rotate=-71.56] at (8.5,-2.4) {$b$};
\draw[->, thick] (8.3,-1.8) -- (8.7,-3);
\draw (9,-3) circle (0.3);
\node at (9,-3) {$q_9^6$};
\node[anchor=south,rotate=-71.56] at (8.5,-3.6) {$a$};
\draw[->, thick] (8.3,-3) -- (8.7,-4.2);
\draw (9,-4.2) circle (0.3);
\draw (9,-4.2) circle (0.25);
\node at (9,-4.2) {$q_9^7$};
\node[anchor=south,rotate=71.56] at (9.5,3.6) {$a$};
\draw[->, thick, gray, dashed] (9.3,3) -- (9.7,4.2);
\draw[thick, gray, dashed] (10,4.2) circle (0.3);
\node at (10,4.2) {$q_{10}^1$};
\node[anchor=south,red] at (9.5,3) {$b$};
\draw[->, thick,red] (9.3,3) -- (9.7,3);
\draw (10,3) circle (0.3);
\node at (10,3) {$q_{10}^2$};
\node[anchor=south] at (9.5,1.8) {$b$};
\draw[->, thick] (9.3,1.8) -- (9.7,1.8);
\draw (10,1.8) circle (0.3);
\node at (10,1.8) {$q_{10}^3$};
\node[anchor=south] at (9.5,0.6) {$a$};
\draw[->, thick] (9.3,0.6) -- (9.7,0.6);
\draw (10,0.6) circle (0.3);
\draw (10,0.6) circle (0.25);
\node at (10,0.6) {$q_{10}^4$};
\node[anchor=south] at (9.5,-0.6) {$b$};
\draw[->, thick] (9.3,-0.6) -- (9.7,-0.6);
\draw (10,-0.6) circle (0.3);
\node at (10,-0.6) {$q_{10}^5$};
\node[anchor=south] at (9.5,-1.8) {$a$};
\draw[->, thick] (9.3,-1.8) -- (9.7,-1.8);
\draw (10,-1.8) circle (0.3);
\draw (10,-1.8) circle (0.25);
\node at (10,-1.8) {$q_{10}^6$};
\node[anchor=south] at (9.5,-3) {$a$};
\draw[->, thick] (9.3,-3) -- (9.7,-3);
\draw (10,-3) circle (0.3);
\draw (10,-3) circle (0.25);
\node at (10,-3) {$q_{10}^7$};
\node[anchor=south] at (9.5,-4.2) {$a$};
\draw[->, thick, gray, dashed] (9.3,-4.2) -- (9.7,-4.2);
\draw[thick, gray, dashed] (10,-4.2) circle (0.3);
\node at (10,-4.2) {$q_{10}^8$};
\node[anchor=south,rotate=-71.56] at (9.5,-4.8) {$b$};
\draw[->, thick] (9.3,-4.2) -- (9.7,-5.4);
\draw (10,-5.4) circle (0.3);
\node at (10,-5.4) {$q_{10}^9$};
\node[anchor=south,rotate=71.56] at (10.5,4.8) {$b$};
\draw[->, thick, gray, dashed] (10.3,4.2) -- (10.7,5.4);
\draw[thick, gray, dashed] (11,5.4) circle (0.3);
\node at (11,5.4) {$q_{11}^1$};
\node[anchor=south,rotate=71.56,red] at (10.5,3.6) {$a$};
\draw[->, thick] (10.3,3) -- (10.7,4.2);
\draw (11,4.2) circle (0.3);
\draw (11,4.2) circle (0.25);
\node at (11,4.2) {$q_{11}^2$};
\node[anchor=south,rotate=71.56] at (10.5,2.4) {$a$};
\draw[->, thick] (10.3,1.8) -- (10.7,3);
\draw (11,3) circle (0.3);
\draw (11,3) circle (0.25);
\node at (11,3) {$q_{11}^3$};
\node[anchor=south,rotate=71.56] at (10.5,1.2) {$a$};
\draw[->, thick, gray, dashed] (10.3,0.6) -- (10.7,1.8);
\draw[thick, gray, dashed] (11,1.8) circle (0.3);
\node at (11,1.8) {$q_{11}^4$};
\node[anchor=south] at (10.5,0.6) {$b$};
\draw[->, thick, gray, dashed] (10.3,0.6) -- (10.7,0.6);
\draw[thick, gray, dashed] (11,0.6) circle (0.3);
\node at (11,0.6) {$q_{11}^5$};
\node[anchor=south] at (10.5,-0.6) {$a$};
\draw[->, thick] (10.3,-0.6) -- (10.7,-0.6);
\draw (11,-0.6) circle (0.3);
\draw (11,-0.6) circle (0.25);
\node at (11,-0.6) {$q_{11}^6$};
\node[anchor=south] at (10.5,-1.8) {$a$};
\draw[->, thick, gray, dashed] (10.3,-1.8) -- (10.7,-1.8);
\draw[thick, gray, dashed] (11,-1.8) circle (0.3);
\node at (11,-1.8) {$q_{11}^{7}$};
\node[anchor=south,rotate=-71.56] at (10.5,-2.4) {$b$};
\draw[->, thick,gray, dashed] (10.3,-1.8) -- (10.7,-3);
\draw[thick, gray, dashed] (11,-3) circle (0.3);
\node at (11,-3) {$q_{11}^{8}$};
\node[anchor=south,rotate=-71.56] at (10.5,-3.6) {$a$};
\draw[->, thick, gray, dashed] (10.3,-3) -- (10.7,-4.2);
\draw[thick, gray, dashed] (11,-4.2) circle (0.3);
\node at (11,-4.2) {$q_{11}^{9}$};
\node[anchor=south,rotate=-71.56] at (10.5,-4.8) {$b$};
\draw[->, thick, gray, dashed] (10.3,-3) -- (10.7,-5.4);
\draw[thick, gray, dashed] (11,-5.4) circle (0.3);
\node at (11,-5.4) {$q_{11}^{10}$};
\node[anchor=south,rotate=-71.56] at (10.5,-6) {$b$};
\draw[->, thick, gray, dashed] (10.3,-4.2) -- (10.7,-6.6);
\draw[thick, gray, dashed] (11,-6.6) circle (0.3);
\node at (11,-6.6) {$q_{11}^{11}$};
\node[anchor=south,rotate=-71.56] at (10.5,-7.2) {$a$};
\draw[->, thick] (10.3,-5.4) -- (10.7,-7.8);
\draw (11,-7.8) circle (0.3);
\draw (11,-7.8) circle (0.25);
\node at (11,-7.8) {$q_{11}^{12}$};
\end{tikzpicture}}

\caption{A trimmed DTA for $\mathcal{L}^{\leq 11}=\lbrace w\in\Sigma^{\leq 11}\vert w \textmd{ is covered by } aba\rbrace$. Dashed states and transitions are those prunned during trimming. A full-size version is available in the appendix}
\end{minipage}
\end{figure}

Since a regular language $\mathcal{L}$ may be infinite, it is natural that any finite automaton $\mathfrak{A}_\mathcal{L}$ associated with such a language must have cycles i.e. $\exists q\in Q,\, w\in \Sigma^*,\, \delta\left(q,\,w\right)=q$ since the number of paths in a \emph{directed acyclic graph} (henceforth \emph{DAG}) is finite. However, any \emph{finite} language admits a trimmed real-time deterministic finite automaton (henceforth \emph{automaton}) whose underlying topology is an (edge-labeled) DAG. For the construction of a minimal DAG associated with a finite language $\mathcal{L}\subseteq\Sigma^{\leq n}$ we refer the reader to the work of Mikov.

More than that, for the case of a finite language we can always use a tree to parse it. We call an automaton with an underlying tree topology a \emph{Tree DFA} (henceforth \emph{DTA}). Note that these are not the bottom-up/top-down DTFAs accepting \emph{tree languages} of Comon et al.~\cite{comon1997tree}.  We instead limit ourselves to the regular \emph{word languages}.

Let $\mathcal{L}\subseteq\Sigma^*$ be a language, $\mathfrak{A}_\mathcal{L}$ be its associated automaton and $n\in\mathbb{N}$ be a finite number. The language $\mathcal{L}^{\leq n}= \lbrace w\in\mathcal{L}\vert \lvert w\rvert \leq n\rbrace$ is finite and thus we can describe it through a DTA. If we have the original automaton we can obtain this DTA through loop unrolling. 

The relation between $\mathfrak{A}_\mathcal{L}$ and $\mathcal{L}^{\leq n}$ is a special one; while $\mathfrak{A}_\mathcal{L}$ can not be isomorphic to any $\mathfrak{A}_{\mathcal{L}^{\leq n}}$, it is nonetheless useful since all $\mathfrak{A}_{\mathcal{L}^{\leq n}}$ have at least as many states as $\mathfrak{A}_{\mathcal{L}}$ (if $n$ is larger than the depth of any final state $q\in F$). Yet, in combination with bounds checking we can recover the finite language. Let us denote by $\Sigma^{\leq n}=\lbrace w\in\Sigma^*\vert \lvert w\rvert \leq n\rbrace=\Sigma^n\Sigma^*$ the language of bounds-conforming words. Since $\mathcal{L}^{\leq n}=\mathcal{L}\cap\Sigma^{\leq n}$, through use of a deterministic finite transducer $\left(q^\prime_i,\,\sigma\right)\rightarrow \left(q^\prime_{i+1},\,\sigma\right)\,\forall i\in\overline{0,\,n-1},\,\sigma\in\Sigma$ with $n$ states ($Q^\prime = \lbrace q^\prime_i \vert i\in\overline{0,\,n}\rbrace$), all final, feeding input into our original $\mathfrak{A}$ we do not suffer from the exponential state boom incurred by any DTA. Thus $\mathfrak{A}_\mathcal{L}$ remains special, and we call any such $\mathfrak{A}$ that can be substituted for it a \emph{covering automaton} for $\mathcal{L}^{\leq n}$. A minimal covering DFA can be obtained in linear time from a given finite language, and it can even be constructed incrementally. 

\subsection{Our Results}
The minimal covering DFA is quite a clever engineering tool but it does not generalize string covers, since it is allowed to contain sub-cycles. For instance, let $\mathcal{L}_{1,\,2}$ be finite languages over disjoint alphabets. Then by serially connecting the two minimal covering DFAs we obtain the minimal covering DFA of $\mathcal{L}_1\mathcal{L}_2$. If we formulate the String Cover problem as a minimal covering DFA problem, then for two strings $w_{1,\,2}$ over disjoint alphabets covered by $s_{1,\,2}$, $s_1s_2$ would be a string cover of $w_1w_2$ which admits \emph{no} (non-trivial) string cover. 

In this paper we take on the quest of generalizing string covers to trees as a natural next level of difficulty while keeping the complexity polynomial. This is motivated by the deep connection between regular expressions and trees. We start by introducing a new class of Cover Automata. In the rest of the paper we use the following notation: for a function $f$, $f(.)$ is the function $f$ over an arbitrary argument.

\begin{definition}[Cover DTA]
\label{def:coverDTA}
Let $\mathfrak{A}=\left(Q,\,\Sigma,\delta,\,q_0,\,F\right)$ be a DTA/NTA. We say that $\mathfrak{A}$ covers a DTA/NFA $\mathfrak{A}^\prime=\left(Q^\prime,\,\Sigma,\delta^\prime,\,q_0^\prime,\,F^\prime\right)$ if and only if there exists a family of functions $\mathbf{\phi}_I=\lbrace\phi_{i\in I}:Q\rightarrow Q^\prime\vert \phi_i\left(\delta\left(.,\,.\right)\right)\subseteq\delta^\prime\left(\phi_i\left(.\right),\,.\right)\rbrace$ such that $\mathbf{\phi}_I\left(Q\right)=\underset{i\in I}\cup\phi_i\left(Q\right)=Q^\prime$ and $\mathbf{\phi}_I\left(F\right)\supseteq F^\prime$.
\end{definition}

 Informally, Definition~\ref{def:coverDTA} states that every state and transition in $\mathfrak{A}^\prime$ belongs to some occurence of $\mathfrak{A}$ in $\mathfrak{A}^\prime$.

\begin{remark*}
It is not required that the occurrence of $\mathfrak{A}$ in $\mathfrak{A}^\prime$ maintains the final states and thus it is not required that $\mathcal{L}\left(\mathfrak{A}\right)\subseteq\mathcal{L}\left(\mathfrak{A}^\prime\right)$. 
\end{remark*}

Notice that the concept that we introduce in this paper does not trivially fit in the expressive hierarchy of automata. There are (finite) languages which admit a non-trivial Cover DTA (i.e., smaller than the one constructed in Theorem~\ref{thm:single_DTA})  but none of whose non-trivial Cover DFAs have a tree topology. Moreover there are (finite) languages which admit no non-trivial Cover DTA. A neat example of languages that do admit non-trivial DFAs is $\mathcal{L}_s^{\leq n}=\lbrace{w\in\Sigma^{\leq n}\vert w \textmd{ is covered by } s}\rbrace$ where $n > \lvert s\rvert$. In fact, it is this particularity which allows us to model String Covers. Thus, we are now ready to introduce our  first result of this paper.

\begin{theorem*}Let $s$ be a word over $\Sigma$, $\lvert s\rvert \le n \in \mathbb{N}^* $ and $\mathfrak{A}$ be the automaton recognizing $s$ and only $s$. Then $\mathfrak{A}$ covers any automaton recognizing a non-empty subset of $\mathcal{L}_s^{\leq n}=\lbrace w\in\Sigma^{\leq n}\vert w \textmd{ is covered by } s\rbrace$
\end{theorem*}

\begin{remark*}
By $\mathbf{\phi}_I\left(F\right)\supseteq F^\prime$, the leaves of $\mathfrak{A}^\prime$'s tree must correspond to final states in $\mathfrak{A}$. If we drop this requirement, then the words accepted by a covered automata may not end in the given string $s$. If we impose a more stringent requirement, say $\mathbf{\phi}_I\left(F\right)=F^\prime$ we cannot recognize some subsets of $\mathcal{L}_s^{\leq n}$, since for any covered word $w$ we are required to cover any prefix $p$ of $w$ that is covered by $s$. We can alternatively replace the requirement with added stringency on transitions, say $\phi_i\left(\delta\left(.,\,.\right)\right)=\delta^\prime\left(\phi_i\left(.\right),\,.\right)$ but then we cannot cover trees branching at multiple points of an ocurrence such as $\lbrace ababa,\,abaaba\rbrace$ with $aba$. Thus, we are at the right degree of constraint.
\end{remark*}

Having constructed such an elaborate model, we check that the model space is itself regular. With a very coarse topology we cannot leverage the usual optimization techniques (e.g., greedy algorithms). Thus, we prove that our relation is at least a partial order.

\begin{theorem*}The Cover relation, $\trianglelefteq$, is a partial ordering over DTA/NTA.
\end{theorem*}

Up to this point, we prove that String Covers can be modeled in a space that is much broader, that of DTA and we show that this space has at least useful algebraic properties. Intuitively, the order which we define brings some benefits over alternative formalisms. 

For example, note that \emph{Approximate String Covers} are not ordered. If $s$ is an $\epsilon$-cover of $t$ and $t$ is a $\delta$-cover of $w$, $s$ is anything between a perfect cover of $w$ and a $\lvert w\rvert$-cover of $w$, i.e. not a cover at all. Consider the case of the periodic sting $a^n$ perfectly covered by $a$ which is $1-covered$ by $b$. $b$ is a $n-cover$ of $a^n$. This is natural for any approximative compression model: with each use of the erroneously compressed token the number of errors increases and thus if the minimal number of $t$ tokens required to cover $w$ is $n$, $s$ is a $n\epsilon-cover$ of $w$.

We now turn our attention towards the computational aspects of this embedding and consider two problems: \emph{decision} and \emph{optimization}. We consider both sequential and parallel models for our calculations.

\begin{theorem*}Deciding whether an automaton $\mathfrak{A}$ covers $\mathfrak{A}^\prime$ takes at most $\mathcal{O}\left(\lvert Q\rvert\lvert Q^\prime\rvert\right)$ serial time or $\mathcal{O}\left(\lvert Q\rvert d\left(\mathfrak{A}^\prime\right)\right)$ massively parallel time.
\end{theorem*}

Simply by using the theorem above we have an algorithm for optimizing arbitrary functions, for which the structure of our space cannot possibly be exploited. For example, we cannot assume even that the function is monotonous with respect to the Cover relation. This yield the bound in Lemma~\ref{lem:trivialminimization}.

A probabilistic interpretation of Lemma~\ref{lem:trivialminimization} is that for the case of a hashing function (presumed to be arbitrary enough) we can obtain as difficult a problem as we want. This could potentially find applicability in \emph{proof-of-work} systems~\cite{biryukov2015proof}. This is formalized in the Theorem~\ref{thm:probabilistic}, restated below.

\begin{theorem*}Let $\mathcal{L}^{\leq n}$ be a finite language, $\mu : DTA\rightarrow \mathbb{R}$ be a partial function, $\mathcal{A}$ be a collection of DTA, $\mu_0$ be a target real number and $R$ be an oracle providing DTA in $\mathcal{A}$ with a consistently significant probability i.e. for which there exists a polynomial $p\left(\mathcal{L}\right)$ such that $P\left[\mu\left(R\left(\right)\right)\leq \mu_0,\,R() \textmd{ covers } \mathfrak{A}_{L^{\leq n}}\right]\geq 1/p\left(\mathcal{L}\right)$. Then computing a DTA $\mathfrak{A}\in\mathcal{A}$ covering $\mathfrak{A}_{\mathcal{L}^{\leq n}}$ takes on average at most $\mathcal{O}\left(p\left(\mathcal{L}\right)\lvert \mathcal{A}\rvert\lvert\mathfrak{A}_{\mathcal{L}^{\leq n}}\rvert^2\right)$ i.e. $\mathcal{O}\left(p\left(\mathcal{L}\right)\lvert \mathcal{A}\rvert \left(\underset{w\in\mathcal{L}^{\leq n}}\sum\lvert w\rvert\right)^2\right)$ serial time or $\mathcal{O}\left(p\left(\mathcal{L}\right)\lvert\mathfrak{A}_{\mathcal{L}^{\leq n}}\rvert\ d\left(\mathfrak{A}_{\mathcal{L}^{\leq n}}\right)\right)$ i.e. $\mathcal{O}\left(p\left(\mathcal{L}\right)\underset{w\in\mathcal{L}^{\leq n}}\sum\lvert w\rvert\underset{w\in\mathcal{L}^{\leq n}}\max\lvert w\rvert\right)$ massively parallel time. 
\end{theorem*}

Another interesting domain of applicability is that where the function itself only considers a polynomially-bounded class of DTA, such as paths. Here, by applying the previous Lemma~\ref{lem:trivialminimization} we obtain a raw, yet extremely  useful result.

\begin{theorem*}Let $\mathcal{L}^{\leq n}$ be a finite language, $\mu : DTA\rightarrow \mathbb{R}$ be a partial function that only takes values over paths and $\mathcal{A}$ be a collection of DTA. Computing $\underset{\mathfrak{A}\in\mathcal{A}}{\arg\max}\lbrace \mu\left(\mathfrak{A}\right)\vert\mathfrak{A}_{\mathcal{L}^{\leq n}} \textmd{ is covered by } \mathfrak{A}\rbrace$ takes at most $\mathcal{O}\left(\lvert\mathcal{L}\rvert \underset{w\in\mathcal{L}^{\leq n}}\min\lvert w\rvert\lvert\mathfrak{A}_{\mathcal{L}^{\leq n}}\rvert\right)$ i.e. $\mathcal{O}\left(\lvert \mathcal{L}\rvert \underset{w\in\mathcal{L}^{\leq n}}\min\lvert w\rvert\underset{w\in\mathcal{L}^{\leq n}}\sum\lvert w\rvert\right)$ serial time or $\mathcal{O}\left(\underset{w\in\mathcal{L}^{\leq n}}\min\lvert w\rvert\ d\left(\mathfrak{A}_{\mathcal{L}^{\leq n}}\right)\right)$ i.e. $\mathcal{O}\left(\underset{w\in\mathcal{L}^{\leq n}}\min\lvert w\rvert\underset{w\in\mathcal{L}^{\leq n}}\max\lvert w\rvert\right)$ massively parallel time. 
\end{theorem*}

Recall that one of the interesting problems modeled by Cover DTA is the Minimal String Cover Problem. This can be seen as an optimization of the length over path automata. Thus, we have an algorithm for the Minimal String Cover Problem and moreover, at the same time, for the Minimal Common String Cover Problem.

\begin{corollary*} Let $\mathcal{L}^{\leq n}$ be a finite language. Computing the shortest cover of all words in $\mathcal{L}$ takes at most $\mathcal{O}\left(\lvert\mathcal{L}\rvert \underset{w\in\mathcal{L}^{\leq n}}\min\lvert w\rvert\lvert\mathfrak{A}_{\mathcal{L}^{\leq n}}\rvert\right)$ i.e. $\mathcal{O}\left(\lvert \mathcal{L}\rvert \underset{w\in\mathcal{L}^{\leq n}}\min\lvert w\rvert\underset{w\in\mathcal{L}^{\leq n}}\sum\lvert w\rvert\right)$ serial time or $\mathcal{O}\left(\underset{w\in\mathcal{L}^{\leq n}}\min\lvert w\rvert\ d\left(\mathfrak{A}_{\mathcal{L}^{\leq n}}\right)\right)$ i.e. $\mathcal{O}\left(\underset{w\in\mathcal{L}^{\leq n}}\min\lvert w\rvert\underset{w\in\mathcal{L}^{\leq n}}\max\lvert w\rvert\right)$ massively parallel time.  
\end{corollary*}

The Minimal Common String Cover Problem has some practical constraints to it, which yield our algorithm optimal for the massively parallel case.

\begin{theorem*} Let $\mathcal{L}^{\leq n}$ be a finite language where all words have roughly the same size i.e. $\exists c\in\mathbb{N}$ a fixed constant such that $c\geq \underset{w\in\mathcal{L}^{\leq n}}\max \lvert w\rvert /\underset{w\in\mathcal{L}^{\leq n}}\min \lvert w\rvert$. Computing the shortest cover of all words in $\mathcal{L}$ takes roughly $\Theta\left(\underset{w\in\mathcal{L}^{\leq n}}\min\lvert w\rvert^2\right)$ massively parallel time.  
\end{theorem*}

Thus in this article we have shown a new embedding of String Covers which generalize them while keeping the problem polynomial in time.

\section{String Covers as Cover DTA}
For the course of this section we  consider fixed an alphabet $\Sigma$, a string $s\in\Sigma^*$ over it, the number $\mathbb{N}\ni n \geq \lvert s\rvert$, the language $\mathcal{L}_s=\lbrace w\in\Sigma^*\vert w \textmd{ is covered by } s\rbrace$ and the family of languages $\mathcal{L}_s^{\leq k}=\lbrace w\in\Sigma^{\leq k}\vert w \textmd{ is covered by } s\rbrace$ with $k\geq \lvert s\rvert$. 

We prove in Theorem~\ref{thm:single_DTA} that for each $\emptyset\subsetneq\mathcal{L}^\prime\subseteq\mathcal{L}_s^{\leq k}$ there exists, up to rooted tree isomorphism, a single DTA accepting it, name it $\mathfrak{A}_{L^\prime}$. For any such language we define its depth as the depth of the underlying tree topology, effectively equal to the length of the longest contained word i.e. $d\left(\mathcal{L}^\prime\right)= \underset{w\in \mathcal{L}^\prime}\max \lvert w\rvert=d\left(\mathfrak{A}_{\mathcal{L}^\prime}\right)$. 

We  show a constructive proof of the existence and uniqueness of the recognizing DTA, $\mathfrak{A}_{\mathcal{L}^\prime}$ and prove the preceding equation in the process. This proof is probably not new and certainly not treacherous, but it is extremely relevant for the proofs of our results. We prove our result by induction and thus we split this proof into the base case (see Lemma~\ref{lem:base_unique_DTA}) and the induction step (see Lemma~\ref{lem:ind_unique_DTA}).

\begin{lemma}
Up to graph isomorphism there exists a single trimmed DFA recognizing a word $w$ and only $w$. This automaton is a trimmed DTA.
\label{lem:base_unique_DTA}
\end{lemma}
\begin{proof}
The existence and uniqueness of the minimal DFA is due to the Myhill-Nerode theorem~\cite{hopcroft2001introduction}. 

We  construct an automaton as follows:
$$Q=\lbrace q_i\vert i\in \overline{0,\,\lvert w\rvert}\rbrace\;\delta\left(q_i,\,w^{i+1}\right)= q_{i+1},\,\forall i\in\overline{0,\,\lvert w\rvert-1}\;F= \lbrace q_{\lvert w\rvert}\rbrace$$

Note that since $F= \lbrace q_{\lvert w\rvert}\rbrace$, $q_{\lvert w\rvert}\not\equiv q_{\lvert w\rvert -1}$. Moreover, 
$$q_{\lvert w\rvert}\not\equiv q_{\lvert w\rvert -1}\Leftrightarrow \delta\left(q_{\lvert w\rvert}-1,\,w^{\lvert w\rvert}\right)\not\equiv \delta\left(q_{\lvert w\rvert -2},\,w^{\lvert w\rvert-1}\right)\Leftrightarrow q_{\lvert w\rvert}-1 \not\equiv q_{\lvert w\rvert}-2\Leftrightarrow \dots$$ 

Thus, our automaton is minimal. Since it has a path topology, if there were another trimmed DFA recognizing this language it would either have one more branch, and thus two final states, and thus accept at least two words, or it would be a longer path, thus accepting a word that is larger than $w$. Thus, there is no other trimmed DFA recognizing this language. Moreover this automaton is a trimmed DTA.
\end{proof}

\begin{lemma}
Let $\mathcal{L}\subseteq\Sigma^{\leq n}$ be a finite language for which there exists a single trimmed DTA, $\mathfrak{A}_\mathcal{L}=\left(Q,\,\Sigma,\,\delta,\,q_0,\,F\right)$ to recognize it, up to rooted tree isomorphism, and $w\in\Sigma^*$. Then the language $\mathcal{L}\cup\lbrace w\rbrace$ also admits a single trimmed DTA to recognize it, up to rooted tree isomorphism.
\label{lem:ind_unique_DTA}
\end{lemma}

\begin{proof}
Firstly, let $p$ be the maximum-length common-prefix of $w$ with any word $w^\prime$ already in $\mathcal{L}$ i.e. $p=\underset{w^\prime\in\mathcal{L}}{\arg\max} \underset{w^i={w^\prime}^i\forall i\in\overline{1,j}}{\underset{j\leq \lvert  w\rvert,\,j\leq \lvert w^\prime\rvert}\max} j$ and $q_i=\delta\left(q_{i-1},\,p^i\right)\forall i\in\overline{1,\,\lvert p\rvert}$. 

If $w=p$, we simply need to make $\delta\left(q_0,\,p\right)=q_{\lvert p\rvert}$ final. If the automaton obtained in this way i.e. $\left(Q,\,\Sigma,\,\delta,\,q_0,\,F\cup\lbrace q_{\lvert p\rvert}\rbrace\right)$ were not unique there would be another automaton, $\mathfrak{A}^\prime$, also recognizing $\mathcal{L}\cup\lbrace w\rbrace$ with either less states, transitions or final states. Since the underlying topology must be a tree, the number of final states is the number of accepted words, and thus that cannot vary and thus it has the same number of maximal directed paths, corresponding to the accepted words. If these differed the language would not be the same and thus these also cannot vary. Hence there exist a single trimmed DTA accepting $\mathcal{L}\cup\lbrace w\rbrace$, up to rooted tree isomorphism.

Otherwise, it $\delta\left(q_{\lvert p\rvert},\,w^{\lvert p\rvert+1}\right)$ does not exist and thus we need to add the remaining states i.e. let $\delta^\prime\vert_{Q\backslash\lbrace q_{\lvert p\rvert}\rbrace\times \Sigma}=\delta,\, \delta^\prime\left(q_{\lvert p\rvert},\,.\right)=\delta\left(q_{\lvert p\rvert},\,.\right), \delta^\prime\left(q_i,\,w^{i+1}\right)=q_{i+1}\forall i\in\overline{0,\,\lvert w\rvert-1}$. Additionally we have to make $\delta\left(q_0,\,w\right)=q_{\lvert w\rvert}$ final. Uniqueness follows \emph{exactly} the same argument as above.
\end{proof}

\begin{theorem}
Let $\mathcal{L}\subseteq\Sigma^{\leq n}$ be a finite language. There exists a single trimmed DTA recognizing it, up to rooted tree isomorphism.
\label{thm:single_DTA}
\end{theorem}

\begin{proof}
Let $\mathcal{L}_{i\in\overline{1,\,\lvert \mathcal{L}\rvert}}$ be a chain of length $\lvert\mathcal{L}\rvert$ over $\mathcal{L}$ ordered by inclusion. We have shown that $\mathcal{L}_1=\lbrace{w_1}\rbrace$ admits a single trimmed DTA recognizing it. Thus, by adding $\mathcal{L}_2\backslash\mathcal{L}_1=\lbrace{w_2}\rbrace$ to it we obtain the language $\mathcal{L}_2$, which must also admit a single trimmed DTA recognizing it. Iteratively doing this $\lvert\mathcal{L}\rvert-2$ more times we obtain that $\mathcal{L}_{\lvert\mathcal{L}\rvert}=\mathcal{L}$ also admits a single trimmed DTA recognizing it, up to rooted tree isomorphism.
\end{proof}

We now prove that the depth of a language $d\left(\mathcal{L}^\prime\right)$ is the \emph{proper} (i.e. smallest) $k$ such that can be truncated to it $ \mathcal{L}^\prime \subseteq \mathcal{L}^{\prime^{\leq k}} $ and show how this is relevant in the context of the language $ \mathcal{L}_s$ of words covered by $s$.

\begin{lemma}
$\emptyset\subsetneq\mathcal{L}^\prime\in\mathcal{L}_s^{\leq k}\Leftrightarrow \emptyset\subsetneq\mathcal{L}^\prime\in\mathcal{L}_s^{\leq d\left(\mathcal{L}^\prime\right)},\,k\geq d\left(\mathcal{L}^\prime\right)$
\end{lemma}

\begin{proof}
Equivalently one may state that:
$$\emptyset\subsetneq\mathcal{L}^\prime\in\mathcal{L}_s\cap\Sigma^{\leq k}\Leftrightarrow \emptyset\subsetneq\mathcal{L}^\prime\in\mathcal{L}_s\cap\Sigma^{\leq d\left(\mathcal{L}^\prime\right)},\,k\geq d\left(\mathcal{L}^\prime\right)$$
Note that $\mathcal{L}^\prime \in \Sigma^{\leq d\left(\mathcal{L}^\prime\right)}$ is trivially true by definition of $d\left(\mathcal{L}^\prime\right)$. Thus the above reduces to:
$$\emptyset\subsetneq\mathcal{L}^\prime\in\mathcal{L}_s\cap\Sigma^{\leq k}\Leftrightarrow \emptyset\subsetneq\mathcal{L}^\prime\in\mathcal{L}_s,\,k\geq d\left(\mathcal{L}^\prime\right)$$
Yet another form of the above is:
$$\left(\emptyset\subsetneq\mathcal{L}^\prime\in\mathcal{L}_s\right)\Rightarrow\left(\mathcal{L}^\prime \in \Sigma^{\leq k}\Leftrightarrow k\geq d\left(\mathcal{L}^\prime\right)\right)$$
However the affirmation $\mathcal{L}^\prime \in\Sigma^{\leq k}\Leftrightarrow k\geq d\left(\mathcal{L}^\prime\right)$ is true by itself. 
$$\mathcal{L}^\prime\in\Sigma^{\leq k}\Leftrightarrow \forall w\in \mathcal{L}^\prime\, w\in\Sigma^{\leq k}\Leftrightarrow \forall w\in \mathcal{L}^\prime\,\lvert w\rvert \leq k \Leftrightarrow \underset{w\in\mathcal{L}^\prime}\max \lvert w\rvert \leq k$$

\end{proof}

\begin{corollary}
Let $\mathcal{L}\subseteq\Sigma^{\leq n}$ be a finite language and $\mathfrak{A}_\mathcal{L}$ be the only trimmed DTA recognizing it. We have $d\left(\mathcal{L}\right)=d\left(\mathfrak{A}_\mathcal{L}\right)$.
\end{corollary}

\begin{proof}
This results from our constructive proof of the existence and uniqueness of the automaton $\mathfrak{A}_\mathcal{L}$. Since maximal paths correspond to accepted words, maximum length paths correspond to maximum length accepted words and thus the depth of the underlying tree is indeed the maximum length of any accepted word in the language.
\end{proof}

\begin{theorem}
Let $s$ be a word over $\Sigma$, $\mathbb{N}^*\ni n \geq \lvert s\rvert$ and $\mathfrak{A}$ be the automaton recognizing $s$ and only $s$. Then $\mathfrak{A}$ covers any automaton recognizing a non-empty subset, $\mathcal{L}^\prime$ of $\mathcal{L}_s^{\leq n}=\lbrace w\in\Sigma^{\leq n}\vert w \textmd{ is covered by } s\rbrace$.
\end{theorem}

\begin{proof}
The proof relies on Defintion~\ref{def:coverDTA}. 

Let $\mathfrak{A}_{\mathcal{L}^\prime}=\left(Q^\prime,\,\Sigma,\,\delta^\prime,\,q^\prime_0,\,F^\prime\right)$ be the trimmed DTA recognizing $\mathcal{L}^\prime$.

Firstly, consider the case $\mathcal{L}^\prime=\lbrace w\rbrace$. We have that $s$ covers $w$ and thus there must exist a set $I\subseteq \overline{1,\,\lvert w\rvert-\lvert s\rvert+1}$ such that $\forall i\in I,\,j\in\overline{1,\,\lvert s\rvert},\,w^{i+j-1}=s^j$ and $\underset{i,j\in\mathcal{I}}\max \lvert i-j\rvert\leq\lvert s\rvert$. Letting $Q=\lbrace q_i\vert i\in\overline{0,\,\lvert s\rvert}$ $Q^\prime=\lbrace q^\prime_i\vert i\in\overline{0,\,\lvert w\rvert}$ we define the family of functions $\phi_{i\in\mathcal{I}}\left(q_j\right)=q^\prime_{i+j-1}$ which satisfies the properties required in the definition of the Covering DTA.

Next, note that since $\mathfrak{A}$ covers all maximal paths $\mathfrak{A}^{f\in F^\prime}_{\mathcal{L}^\prime}$ of $\mathfrak{A}_{\mathcal{L}^\prime}$ by concatenating the associated families of functions $\mathbf{\phi}^f$ we obtain a larger family $\mathbf{\phi}=\underset{f\in F^\prime}\cup\mathbf{\phi}^f$ which satisfies the properties required in the definition of the Covering DTA. 

Thus, we have that $\mathfrak{A}$ indeed covers any such $\mathfrak{A}_{\mathcal{L}^\prime}$.

\end{proof}

\section{A Topology over Cover DTA}
We have thus far proposed a candidate for the generalization of string covers. Yet, one of the fundamental properties of string covering is that it is a partial order. Any string covers itself and the cover of a cover of a string is a cover of that string. In this section we check if our candidate cover relation exhibits such a topology. We check whether a DTA covers itself and whether the cover of a cover of a DFA is a cover of that DFA.

\begin{remark*}
Formally, covering can be seen as a relation i.e. $\trianglelefteq\,= \lbrace \left(\mathfrak{A}^\prime,\,\mathfrak{A}\right)\in {DTA}\times{DFA} \vert \mathfrak{A}^\prime \textmd{ covers } \mathfrak{A}\rbrace$
\end{remark*}

\begin{lemma}
Let $\mathfrak{A}^j=\left(Q^j,\,\Sigma,\,\delta^j,\,q_0^j,\,F^j\right),\, j\in\overline{1,\,2,\,3}$ be such that $A^1\trianglelefteq \mathfrak{A}^2\trianglelefteq \mathfrak{A}^3$. Then $\mathfrak{A}^1 \trianglelefteq \mathfrak{A}^3$. 
\end{lemma}
\begin{proof}
Since $\mathfrak{A}^1\trianglelefteq \mathfrak{A}^2\trianglelefteq \mathfrak{A}^3$ there exist the families $$\mathbf{\phi}_{I^j}=\lbrace \phi_{i\in I^j}:\mathfrak{A}^j\rightarrow\mathfrak{A}^{j+1}\vert \phi^j\left(\delta^{j}\left(.,\,.\right)\right)\subseteq\delta^{j+1}\left(\phi^j\left(.\right),\,.\right)\rbrace;\;\mathbf{\phi}_{I^j}\left(Q^j\right)=Q^{j+1}$$

Consider the Minkowski composition: 
$$\phi_{I^1\times I^2}=\mathbf{\phi}_{I^2}\circ\mathbf{\phi}_{I^1}=\lbrace\phi_{i^1\in I^1,\,i^2\in I^2}= \phi_{i^2}\circ \phi_{i^1}\rbrace$$

Thus it must be that
$$\phi_{i^2}\left(\phi_{i^1}\left(\delta^1\left(.,\,.\right)\right)\right)\subseteq\phi_{i^2}\left(\delta^2\left(\phi_{i^1}\left(.\right),\,.\right)\right)\subseteq \delta^3\left(\phi_{i^2}\left(\phi_{i^1}\left(.\right)\right),\,.\right)$$
which in accord with our naming convention translates to:
$$\mathbf{\phi}_{i^1,\,i^2}\left(\delta^1\left(.,\,.\right)\right)\subseteq\delta^3\left(\phi_{i^1,\,i^2}\left(.,\,.\right),\,.\right)$$

Moreover, we have that
$$\mathbf{\phi}_{I^1\times I^2}\left(Q^1\right)=\underset{i^2\in I^2}\cup\underset{i^1\in I^1}\cup\phi_{i^2}\left(\phi_{i^1}\left(Q^1\right)\right)=\underset{i^2\in I^2}\cup\phi_{i^2}\left(Q^2\right)=Q^3$$
$$\mathbf{\phi}_{I^1\times I^2}\left(F^1\right)=\underset{i^2\in I^2}\cup\underset{i^1\in I^1}\cup\phi_{i^2}\left(\phi_{i^1}\left(F^1\right)\right)\supseteq\underset{i^2\in I^2}\cup\phi_{i^2}\left(F^2\right)\supseteq F^3$$

\end{proof}

\begin{theorem}
The Cover relation, $\trianglelefteq$, is a partial ordering over DTA/NTA.
\end{theorem}
\begin{proof}
Via the family composed of the identity function, any DTA/NTA covers itself and thus $\trianglelefteq$ is reflexive. Since it is also transitive, it is indeed an ordering over DTA/NTA.
\end{proof}

\section{Cover DTA Recognition}
In this section we present a Cover DTA Recognition algorithm based on Message Passing. Like in the previous sections we  consider a (presumably small) automaton $\mathfrak{A}=\left(Q,\,\Sigma,\,\delta,\,q_0,\,F\right)$ and a (presumably bigger) automaton $\mathfrak{A}^\prime=\left(Q^\prime,\,\Sigma,\,\delta^\prime,\,q^\prime_0,\,F^\prime\right)$ and we would like to answer whether $\mathfrak{A}$ covers $\mathfrak{A}^\prime$. An instance of this problem is defined by the pair $\left(\mathfrak{A},\,\mathfrak{A}^\prime\right)$.

\begin{algorithm}
\caption{Cover DTA Recognition}
\label{alg:recognition}
\begin{algorithmic}[1]
\Procedure {Node}{$q^\prime$, $\mathfrak{A}$, $\mathfrak{A}^\prime$}\\
 \hspace*{\algorithmicindent} \textbf{Output:} $A_{q^\prime}\subseteq Q$ the states in $\mathfrak{A}$ that $q^\prime$ can be covered by 
\ForAll {$q^\prime_\sigma\in \delta^\prime\left(q^\prime,\,.\right)$}
\State wait for a message $A_\sigma$ from $q^\prime_\sigma$
\EndFor
\State $A_{q^\prime}\leftarrow F \cup \lbrace A^i_{q^\prime}\in Q\backslash F\vert \forall q^\prime_\sigma\,\exists A_\sigma^i\in A_\sigma\,\delta\left(A_{q^\prime}^i,\,\sigma\right)=A_{\sigma}^i \rbrace$
\State $Q^\prime_{prec}\leftarrow \delta^{\prime^{-1}}\left(q^\prime\right)$
\If {$Q^\prime_{prec} = \emptyset$}
\If {$q_0\notin A_{q^\prime}$}
\State \Return $\emptyset$ (broadcasting it to all $q^\prime_\sigma$)
\Else
\State $A_{q^\prime}\leftarrow \lbrace q_0\rbrace$
\EndIf
\Else $\ Q^\prime_{prec}=\lbrace \left(q^\prime_{prec},\,\sigma\right)\rbrace$
\State pass $A_{q^\prime}$ to $q_{prec}^\prime$ and wait for an answer
\State $A_{q^\prime}\leftarrow \delta\left({\sc Node}\left(q_{prec},\mathfrak{A},\,\mathfrak{A}^\prime\right),\,\sigma\right)\cup\left(A_{q^\prime}\cap \lbrace q_0\rbrace\right)$
\EndIf
\If {$q^\prime \in F^\prime$ and $A_{q^\prime}\cap F=\emptyset$}
\State \Return $\emptyset$ (broadcasting it to all $q^\prime_\sigma$)
\EndIf
\State \Return $A_{q^\prime}$ (broadcasting it to all $q^\prime_\sigma$)
\EndProcedure
\end{algorithmic}
\end{algorithm}

\begin{theorem}Deciding whether an automaton $\mathfrak{A}$ covers $\mathfrak{A}^\prime$ takes at most $\mathcal{O}\left(\lvert Q\rvert\lvert Q^\prime\rvert\right)$ serial time or $\mathcal{O}\left(\lvert Q\rvert d\left(\mathfrak{A}^\prime\right)\right)$ massively parallel time.
\end{theorem}

\begin{proof}
In order to do this we  place agents in the nodes corresponding to the states $Q^\prime$, denoting both them and their corresponding states by $\lbrace q^\prime_i\vert i\in\overline{1,\,\lvert Q^\prime\rvert}\rbrace$. For the rest of the section we  describe their communication protocol and how they can achieve the identification of the functions $\phi_i$. 

In order to find the functions $\phi_i$ it is sufficient to ask which agents can fulfill the role of $q_0$ since that is the most demanding node (i.e. whether $\phi_i\left(q_0\right)$ can possibly be $q^\prime$). By this logic, it seems reasonable that the least demanding roles be the leaves, and consequently that their role be played by the leaves of $\mathfrak{A}^\prime$. Indeed, the leaves of $\mathfrak{A}^\prime$ must respect $\phi\left(\delta\left(q,\,.\right)\right)\subseteq\delta^\prime\left(\phi\left(q\right),\,.\right)=\emptyset$ i.e. they can only ever be leaves of $\mathfrak{A}$. Thus these leaves begin by sending a message to their predecessors that they can fulfill any role in $F$. Note that they may be in any superposition of final states i.e. in $\mathbf{Q}=\lbrace Q^i\vert Q^i \subseteq F\rbrace$.

A generic node $q^\prime$ can ever play the role of $q$ iff $\phi\left(\delta\left(q,\,.\right)\right)\subseteq\delta^\prime\left(\phi\left(q\right),\,.\right)=\delta^\prime\left(q^\prime,\,.\right)$ i.e. their successors may only play the roles of the successors of $q$, no others. Thus, a node must firstly wait for its successors to communicate their availabilities to it. Notably, nothing special is required for a node to play the role of a final state and this node can be played simultaneously with another. 

Let us assume that the children $\lbrace q^\prime_\sigma =\delta\left(q^\prime, \sigma\right)\rbrace$ of a generic node $q^\prime$ can be in the superpositions $\mathbf{Q}_\sigma=\lbrace Q^i_{\sigma}= \lbrace q^j_{\sigma,\,i} \rbrace \subseteq Q \rbrace$. We can now define the eligible superpositions of this state to be: 

\vspace*{-1.5em}
$$\mathbf{Q}= \lbrace Q^i\subseteq Q\vert \forall q\in Q^i\backslash F\,\forall q^\prime_\sigma\,\exists Q_\sigma^i\in\mathbf{Q}_\sigma\,\forall q_{\sigma,\,i}^j\,\delta\left(q,\,\sigma\right)=q_{\sigma,\,i}^j\rbrace$$

\vspace*{-0.5em}
Notably, this formula works for leaves too. Moreover, $\mathbf{Q}$ is closed under intersection and union and hence we could simply send its atoms:

\vspace*{-1.5em}
$$A_{q^\prime}= F \cup \lbrace A^i_{q^\prime}\in Q\backslash F\vert \forall q^\prime_\sigma\,\exists Q_\sigma^i\in\mathbf{Q}_\sigma\,\forall q_{\sigma,\,i}^j\,\delta\left(A_{q^\prime}^i,\,\sigma\right)=q_{\sigma,\,i}^j \rbrace$$

\vspace*{-0.5em}
Yet if instead of $\mathbf{Q}_\sigma$ we received its atoms $A_\sigma$, those would suffice in deliberating the existence of an adequate $Q_\sigma^i$, since if any non-empty such $Q_\sigma^i$ exists one of its atoms can work in its stead. 

\vspace*{-1.5em}
$$A_{q^\prime}= F \cup \lbrace A^i_{q^\prime}\in Q\backslash F\vert \forall q^\prime_\sigma\,\exists A_\sigma^i\in A_\sigma\,\delta\left(A_{q^\prime}^i,\,\sigma\right)=A_{\sigma}^i \rbrace$$

\vspace*{-0.5em}
Note that this formula is linear in size. Thus transmission takes $\lvert \delta\left(q^\prime,\,.\right)\rvert \mathcal{O}\left(\lvert Q\backslash F\rvert\right)$ if we omit the redundant information and only ever actually transmit $A_{q^\prime}\backslash F$. Getting all this information across the network takes $\left(d\left(\mathfrak{A}^\prime\right)+w\left(\mathfrak{A}^\prime\right)\right)\mathcal{O}\left(\lvert Q\backslash F\rvert \right)$ time if agents are actually independent, up to $d\left(\mathfrak{A}^\prime\right)\mathcal{O}\left(\lvert Q\backslash F\rvert \right)$ if they are sufficiently parallel, or $\lvert Q^\prime\rvert\mathcal{O}\left(\lvert Q\backslash F\rvert \right)$ if it is serialized according to a precomputed topological sorting over $\mathfrak{A}^\prime$.

These are the roles that can be played, but the roles are \emph{actually} played only when requested by a predecessor. Thus, we must check whether while trailing any path in $\mathfrak{A}^\prime$ we are only traversing actual occurrence of $\mathfrak{A}$.

Note that the source state $q^\prime_0$ must play the role of $q_0$. If it cannot ($q_0\notin A_{q^\prime_0}$) then $\mathfrak{A}$ does not cover $\mathfrak{A}^\prime$. Assuming henceforth that it can, it collapses $A_{q^\prime_0}$ to $\lbrace q_0\rbrace$. Moreover, for each of its children it prunes $A_\sigma$ to $\delta\left(q^\prime_0,\,\sigma\right)\cup\left(A_\sigma\cap\lbrace q_0\rbrace\right)$ and sends it back to them and shuts down. If ever this quantity is $\emptyset$ the algorithm fails (delivers a negative answer).

Next, for a generic node $q^\prime$ who receives back its now pruned $A_{q^\prime}$ it performs \emph{all} roles it can. Thus it prunes the $A_\sigma$ of its children to: 

\vspace*{-1.5em}
$$A_\sigma \leftarrow \lbrace \delta\left(A_{q^\prime}^i,\,\sigma\right)\cup\left(A_\sigma\cap\lbrace q_0\rbrace\right)\vert A_{q^\prime}^i\in A_{q^\prime}\rbrace$$

\vspace*{-0.5em}
and sends it back to them and shuts down. If ever this quantity is $\emptyset$ or if $q^\prime\in F^\prime$ but $A_{q^\prime}\cap F=\emptyset$ the algorithm fails (delivers a negative answer). Notably, this procedure also applies to $q^\prime_0$ after it has pruned its own $A_{q^\prime_0}$. The second phase of the algorithm has the same running time in all situations.

If all the agents shut down successfully then the flow traveled through $\mathfrak{A}$ only through occurrences of $\mathfrak{A}^\prime$ and ended in final states. The second part ensures that the functions $\phi_i$ are indeed completely defined. In this case, since each node has waited the answer from all successors at the first stage it is guaranteed that $\mathbf{\phi}_I\left(Q\right)=Q^\prime$. Note that $A_{q^\prime} = \mathbf{\phi}_I^{-1}\left(q^\prime\right)$.

In conclusion checking whether $\mathfrak{A}^\prime$ covers $\mathfrak{A}$ takes time at most $\mathcal{O}\left(\lvert Q\rvert\lvert Q^\prime\rvert\right)$ for a serial implementation and $\mathcal{O}\left(\lvert Q\rvert d \left(\mathfrak{A}^\prime\right)\right)$ for a massively parallel one. 
\end{proof}

\begin{remark*}
The above works even if $A_\sigma$ are not pruned by their ancestor, but by the agents themselves, such that any node only broadcasts its $A_{q^\prime}$. This observation leads to the alternative formulation in Algorithm~\ref{alg:recognition}.
\end{remark*}

\section{Cover DTA Minimization}

Let $\mathcal{L}^{\leq n}$ be a finite language, $\mathfrak{A}_{\mathcal{L}^{\leq n}}$ its associated DTA, $\mu : DTA\rightarrow \mathbb{R}$ be a partial function and $\mathcal{A}$ be a collection of DTA. For the remainder of this section we attempt to minimize $\mu$ over $\mathcal{A}$ i.e. compute the DTA:
$$\mathfrak{A}^*=\underset{\mathfrak{A}\in\mathcal{A}}{\arg\max}\lbrace \mu\left(\mathfrak{A}\right)\vert\mathfrak{A}_{\mathcal{L}^{\leq n}} \textmd{ is covered by } \mathfrak{A}\rbrace$$
\vspace*{-3em}
\begin{lemma}
\label{lem:trivialminimization}
Computing $\mathfrak{A}^*$ takes at most $\mathcal{O}\left(\lvert \mathcal{A}\rvert\lvert\mathfrak{A}_{\mathcal{L}^{\leq n}}\rvert^2\right)$ i.e. $\mathcal{O}\left(\lvert \mathcal{A}\rvert \left(\underset{w\in\mathcal{L}^{\leq n}}\sum\lvert w\rvert\right)^2\right)$ serial time or $\mathcal{O}\left(\lvert\mathfrak{A}_{\mathcal{L}^{\leq n}}\rvert\ d\left(\mathfrak{A}_{\mathcal{L}^{\leq n}}\right)\right)$ i.e. $\mathcal{O}\left(\underset{w\in\mathcal{L}^{\leq n}}\sum\lvert w\rvert\underset{w\in\mathcal{L}^{\leq n}}\max\lvert w\rvert\right)$ massively parallel time. 
\end{lemma}

\begin{proof}
Sequentially or in parallel check for every $\mathfrak{A}\in\mathcal{A}$ if it does indeed cover $\mathfrak{A}_{\mathcal{L}^{\leq n}}$. Note that only those satisfying $\lvert \mathfrak{A}\rvert \leq \lvert \mathfrak{A}_{\mathcal{L}^{\leq n}}\rvert$ i.e. with a smaller number of states need to be checked because otherwise not even one $\phi_i$ can be defined. In the massively parallel case, after the network has been set up computing the minimum of $\mu$ for admissible $\mathfrak{A}$ takes linear time by broadcasting.
\end{proof}

\begin{theorem}
\label{thm:probabilistic}
Let $\mu_0$ be a target real number and $R$ be an oracle providing DTA in $\mathcal{A}$ with a consistently significant probability i.e. for which there exists a polynomial $p\left(\mathcal{L}\right)$ such that $P\left[\mu\left(R\left(\right)\right)\leq \mu_0,\,\mathfrak{A} \textmd{ covers } \mathfrak{A}_{L^{\leq n}}\right]\geq 1/p\left(\mathcal{L}\right)$. Then computing a DTA $\mathfrak{A}\in\mathcal{A}$ covering $\mathfrak{A}_{\mathcal{L}^{\leq n}}$ takes on average at most $\mathcal{O}\left(p\left(\mathcal{L}\right)\lvert \mathcal{A}\rvert\lvert\mathfrak{A}_{\mathcal{L}^{\leq n}}\rvert^2\right)$ i.e. $\mathcal{O}\left(p\left(\mathcal{L}\right)\lvert \mathcal{A}\rvert \left(\underset{w\in\mathcal{L}^{\leq n}}\sum\lvert w\rvert\right)^2\right)$ serial time or $\mathcal{O}\left(p\left(\mathcal{L}\right)\lvert\mathfrak{A}_{\mathcal{L}^{\leq n}}\rvert\ d\left(\mathfrak{A}_{\mathcal{L}^{\leq n}}\right)\right)$ i.e. $\mathcal{O}\left(p\left(\mathcal{L}\right)\underset{w\in\mathcal{L}^{\leq n}}\sum\lvert w\rvert\underset{w\in\mathcal{L}^{\leq n}}\max\lvert w\rvert\right)$ massively parallel time. 
\end{theorem}

\begin{proof}
While no suitable $\mathfrak{A}$ has been found get a candidate from $R$ and check it. Using the recognition bounds and considering the number of attempts required on average for hitting the target value is on average $1/P\left[\mu\left(R\left(\right)\right)\leq \mu_0,\,\mathfrak{A} \textmd{ covers } \mathfrak{A}_{L^{\leq n}}\right]\leq p\left(\mathcal{L}\right)$ the result follows trivially.
\end{proof}

\begin{theorem}If we know a priori that $\mu$ only takes values over paths, computing $\mathfrak{A}^*$ takes at most $\mathcal{O}\left(\lvert\mathcal{L}\rvert \underset{w\in\mathcal{L}^{\leq n}}\min\lvert w\rvert\lvert\mathfrak{A}_{\mathcal{L}^{\leq n}}\rvert\right)$ i.e. $\mathcal{O}\left(\lvert \mathcal{L}\rvert \underset{w\in\mathcal{L}^{\leq n}}\min\lvert w\rvert\underset{w\in\mathcal{L}^{\leq n}}\sum\lvert w\rvert\right)$ serial time or $\mathcal{O}\left(\underset{w\in\mathcal{L}^{\leq n}}\min\lvert w\rvert\ d\left(\mathfrak{A}_{\mathcal{L}^{\leq n}}\right)\right)$ i.e. $\mathcal{O}\left(\underset{w\in\mathcal{L}^{\leq n}}\min\lvert w\rvert\underset{w\in\mathcal{L}^{\leq n}}\max\lvert w\rvert\right)$ massively parallel time. 
\end{theorem}

\begin{proof}
Any admissible $\mathfrak{A}$ must be a partial path in $\mathfrak{A}_{\mathcal{L}^{\leq n}}$, starting with $q_0$ and which are at most the number of states. Moreover, any such path must end in a final state, which are exactly as many as words in $\mathcal{L}$. Thus we can create a set $\mathcal{A}$ of DTA that is exactly $\lvert \mathcal{L}\rvert$ big. Note that these path automata have at most $\underset{w\in\mathcal{L}^{\leq n}}\min\lvert w\rvert$ states or else $\mathbf{\phi}_I\left(F\right)\not\supseteq F$. Applying the cover recognition bound the result follows.
\end{proof}

\begin{corollary} Computing the shortest cover of all words in $\mathcal{L}$ takes at most $\mathcal{O}\left(\lvert\mathcal{L}\rvert \underset{w\in\mathcal{L}^{\leq n}}\min\lvert w\rvert\lvert\mathfrak{A}_{\mathcal{L}^{\leq n}}\rvert\right)$ i.e. $\mathcal{O}\left(\lvert \mathcal{L}\rvert \underset{w\in\mathcal{L}^{\leq n}}\min\lvert w\rvert\underset{w\in\mathcal{L}^{\leq n}}\sum\lvert w\rvert\right)$ serial time or $\mathcal{O}\left(\underset{w\in\mathcal{L}^{\leq n}}\min\lvert w\rvert\ d\left(\mathfrak{A}_{\mathcal{L}^{\leq n}}\right)\right)$ i.e. $\mathcal{O}\left(\underset{w\in\mathcal{L}^{\leq n}}\min\lvert w\rvert\underset{w\in\mathcal{L}^{\leq n}}\max\lvert w\rvert\right)$ massively parallel time. 
\end{corollary}

\begin{proof}
A set of strings admits a common cover iff they can be covered by a common string iff their corresponding DTA can be covered by its DTA iff the big DTA can be covered by its DTA. The tree topology does not come into play because the cover should be a path. Thus we simply need to optimize the depth over paths. Applying the theorem above the result follows immediately.
\end{proof}

\begin{lemma} Computing the shortest cover of all words in $\mathcal{L}$ takes at least $\Omega\left(\underset{w\in\mathcal{L}^{\leq n}}\min\lvert w\rvert\underset{w\in\mathcal{L}^{\leq n}}\sum\lvert w\rvert\right)$ serial time or $\Omega\left(\underset{w\in\mathcal{L}^{\leq n}}\min\lvert w\rvert^2\right)$ massively parallel time.  
\end{lemma}
\begin{proof}
For all words we must, in the worst case, compute the covers of length up to $\underset{w\in\mathcal{L}^{\leq n}}\min\lvert w\rvert$ or equivalently check those resulting from a prior step. Checking is at least linear in $\lvert w\rvert$. Even if we knew where to check and put a grid to work we couldn't do better. 
\end{proof}
\begin{theorem} Let all words have roughly the same size i.e. $\exists c\in\mathbb{N}$ a fixed constant such that $c\geq \underset{w\in\mathcal{L}^{\leq n}}\max \lvert w\rvert /\underset{w\in\mathcal{L}^{\leq n}}\min \lvert w\rvert$. Computing the shortest cover of all words in $\mathcal{L}$ takes roughly $\Theta\left(\underset{w\in\mathcal{L}^{\leq n}}\min\lvert w\rvert^2\right)$ massively parallel time.  
\end{theorem}
\begin{proof}
$\mathcal{O}\left(\underset{w\in\mathcal{L}^{\leq n}}\min\lvert w\rvert\underset{w\in\mathcal{L}^{\leq n}}\max\lvert w\rvert\right)\subseteq\mathcal{O}\left(\underset{w\in\mathcal{L}^{\leq n}}\min\lvert w\rvert c \underset{w\in\mathcal{L}^{\leq n}}\min\lvert w\rvert\right)=\mathcal{O}\left(\underset{w\in\mathcal{L}^{\leq n}}\min\lvert w\rvert^2\right)$
Applying the lemma above the result follows immediately.
\end{proof}

\bibliography{bibliography}

\newpage

\appendix

\begin{figure}
\begin{center}
\begin{tikzpicture}
\draw[->, thick] (-0.7,0) -- (-0.3,0);
\draw (0,0) circle (0.3);
\node at (0,0) {$q_0$};
\node[anchor=south,red] at (0.5,0) {$a$};
\draw[->, thick,red] (0.3,0) -- (0.7,0);
\draw (1,0) circle (0.3);
\node at (1,0) {$q_1$};
\node[anchor=south,red] at (1.5,0) {$b$};
\draw[->, thick,red] (1.3,0) -- (1.7,0);
\draw (2,0) circle (0.3);
\node at (2,0) {$q_2$};
\node[anchor=south,red] at (2.5,0) {$a$};
\draw[->, thick,red] (2.3,0) -- (2.7,0);
\draw (3,0) circle (0.3);
\draw (3,0) circle (0.25);
\node at (3,0) {$q_3$};
\node[anchor=south, rotate=36.83,red] at (3.5,0.3) {$a$};
\draw[->, thick,red] (3.3,0) -- (3.7,0.6);
\draw (4,0.6) circle (0.3);
\node at (4,0.6) {$q_4^1$};
\node[anchor=south, rotate=-36.83] at (3.5,-0.3) {$b$};
\draw[->, thick] (3.3,0) -- (3.7,-0.6);
\draw (4,-0.6) circle (0.3);
\node at (4,-0.6) {$q_4^2$};
\node[anchor=south,red] at (4.5,0.6) {$b$};
\draw[->, thick,red] (4.3,0.6) -- (4.7,0.6);
\draw (5,0.6) circle (0.3);
\node at (5,0.6) {$q_5^1$};
\node[anchor=south] at (4.5,-0.6) {$a$};
\draw[->, thick] (4.3,-0.6) -- (4.7,-0.6);
\draw (5,-0.6) circle (0.25);
\draw (5,-0.6) circle (0.3);
\node at (5,-0.6) {$q_5^2$};
\node[anchor=south,red] at (5.5,0.6) {$a$};
\draw[->, thick,red] (5.3,0.6) -- (5.7,0.6);
\draw (6,0.6) circle (0.3);
\draw (6,0.6) circle (0.25);
\node at (6,0.6) {$q_6^1$};
\node[anchor=south] at (5.5,-0.6) {$a$};
\draw[->, thick] (5.3,-0.6) -- (5.7,-0.6);
\draw (6,-0.6) circle (0.3);
\node at (6,-0.6) {$q_6^2$};
\node[anchor=south,rotate=-71.56] at (5.5,-1.2) {$b$};
\draw[->, thick] (5.3,-0.6) -- (5.7,-1.8);
\draw (6,-1.8) circle (0.3);
\node at (6,-1.8) {$q_6^3$};
\node[anchor=south,rotate=71.56,red] at (6.5,1.2) {$a$};
\draw[->, thick,red] (6.3,0.6) -- (6.7,1.8);
\draw (7,1.8) circle (0.3);
\node at (7,1.8) {$q_7^1$};
\node[anchor=south] at (6.5,0.6) {$b$};
\draw[->, thick] (6.3,0.6) -- (6.7,0.6);
\draw (7,0.6) circle (0.3);
\node at (7,0.6) {$q_7^2$};
\node[anchor=south] at (6.5,-0.6) {$b$};
\draw[->, thick] (6.3,-0.6) -- (6.7,-0.6);
\draw (7,-0.6) circle (0.3);
\node at (7,-0.6) {$q_7^3$};
\node[anchor=south] at (6.5,-1.8) {$a$};
\draw[->, thick] (6.3,-1.8) -- (6.7,-1.8);
\draw (7,-1.8) circle (0.3);
\draw (7,-1.8) circle (0.25);
\node at (7,-1.8) {$q_7^4$};
\node[anchor=south,red] at (7.5,1.8) {$b$};
\draw[->, thick,red] (7.3,1.8) -- (7.7,1.8);
\draw (8,1.8) circle (0.3);
\node at (8,1.8) {$q_8^1$};
\node[anchor=south] at (7.5,0.6) {$a$};
\draw[->, thick] (7.3,0.6) -- (7.7,0.6);
\draw (8,0.6) circle (0.3);
\draw (8,0.6) circle (0.25);
\node at (8,0.6) {$q_8^2$};
\node[anchor=south] at (7.5,-0.6) {$a$};
\draw[->, thick] (7.3,-0.6) -- (7.7,-0.6);
\draw (8,-0.6) circle (0.3);
\draw (8,-0.6) circle (0.25);
\node at (8,-0.6) {$q_8^3$};
\node[anchor=south] at (7.5,-1.8) {$a$};
\draw[->, thick] (7.3,-1.8) -- (7.7,-1.8);
\draw (8,-1.8) circle (0.3);
\node at (8,-1.8) {$q_8^4$};
\node[anchor=south,rotate=-71.56] at (7.5,-2.4) {$b$};
\draw[->, thick] (7.3,-1.8) -- (7.7,-3);
\draw (8,-3) circle (0.3);
\node at (8,-3) {$q_8^5$};
\node[anchor=south,rotate=71.56,red] at (8.5,2.4) {$a$};
\draw[->, thick,red] (8.3,1.8) -- (8.7,3);
\draw (9,3) circle (0.3);
\draw (9,3) circle (0.25);
\node at (9,3) {$q_9^1$};
\node[anchor=south,rotate=71.56] at (8.5,1.2) {$a$};
\draw[->, thick] (8.3,0.6) -- (8.7,1.8);
\draw (9,1.8) circle (0.3);
\node at (9,1.8) {$q_9^2$};
\node[anchor=south] at (8.5,0.6) {$b$};
\draw[->, thick] (8.3,0.6) -- (8.7,0.6);
\draw (9,0.6) circle (0.3);
\node at (9,0.6) {$q_9^3$};
\node[anchor=south] at (8.5,-0.6) {$a$};
\draw[->, thick] (8.3,-0.6) -- (8.7,-0.6);
\draw (9,-0.6) circle (0.3);
\node at (9,-0.6) {$q_9^4$};
\node[anchor=south,rotate=-71.56] at (8.5,-1.2) {$b$};
\draw[->, thick] (8.3,-0.6) -- (8.7,-1.8);
\draw (9,-1.8) circle (0.3);
\node at (9,-1.8) {$q_9^5$};
\node[anchor=south,rotate=-71.56] at (8.5,-2.4) {$b$};
\draw[->, thick] (8.3,-1.8) -- (8.7,-3);
\draw (9,-3) circle (0.3);
\node at (9,-3) {$q_9^6$};
\node[anchor=south,rotate=-71.56] at (8.5,-3.6) {$a$};
\draw[->, thick] (8.3,-3) -- (8.7,-4.2);
\draw (9,-4.2) circle (0.3);
\draw (9,-4.2) circle (0.25);
\node at (9,-4.2) {$q_9^7$};
\node[anchor=south,rotate=71.56] at (9.5,3.6) {$a$};
\draw[->, thick, gray, dashed] (9.3,3) -- (9.7,4.2);
\draw[thick, gray, dashed] (10,4.2) circle (0.3);
\node at (10,4.2) {$q_{10}^1$};
\node[anchor=south,red] at (9.5,3) {$b$};
\draw[->, thick,red] (9.3,3) -- (9.7,3);
\draw (10,3) circle (0.3);
\node at (10,3) {$q_{10}^2$};
\node[anchor=south] at (9.5,1.8) {$b$};
\draw[->, thick] (9.3,1.8) -- (9.7,1.8);
\draw (10,1.8) circle (0.3);
\node at (10,1.8) {$q_{10}^3$};
\node[anchor=south] at (9.5,0.6) {$a$};
\draw[->, thick] (9.3,0.6) -- (9.7,0.6);
\draw (10,0.6) circle (0.3);
\draw (10,0.6) circle (0.25);
\node at (10,0.6) {$q_{10}^4$};
\node[anchor=south] at (9.5,-0.6) {$b$};
\draw[->, thick] (9.3,-0.6) -- (9.7,-0.6);
\draw (10,-0.6) circle (0.3);
\node at (10,-0.6) {$q_{10}^5$};
\node[anchor=south] at (9.5,-1.8) {$a$};
\draw[->, thick] (9.3,-1.8) -- (9.7,-1.8);
\draw (10,-1.8) circle (0.3);
\draw (10,-1.8) circle (0.25);
\node at (10,-1.8) {$q_{10}^6$};
\node[anchor=south] at (9.5,-3) {$a$};
\draw[->, thick] (9.3,-3) -- (9.7,-3);
\draw (10,-3) circle (0.3);
\draw (10,-3) circle (0.25);
\node at (10,-3) {$q_{10}^7$};
\node[anchor=south] at (9.5,-4.2) {$a$};
\draw[->, thick, gray, dashed] (9.3,-4.2) -- (9.7,-4.2);
\draw[thick, gray, dashed] (10,-4.2) circle (0.3);
\node at (10,-4.2) {$q_{10}^8$};
\node[anchor=south,rotate=-71.56] at (9.5,-4.8) {$b$};
\draw[->, thick] (9.3,-4.2) -- (9.7,-5.4);
\draw (10,-5.4) circle (0.3);
\node at (10,-5.4) {$q_{10}^9$};
\node[anchor=south,rotate=71.56] at (10.5,4.8) {$b$};
\draw[->, thick, gray, dashed] (10.3,4.2) -- (10.7,5.4);
\draw[thick, gray, dashed] (11,5.4) circle (0.3);
\node at (11,5.4) {$q_{11}^1$};
\node[anchor=south,rotate=71.56,red] at (10.5,3.6) {$a$};
\draw[->, thick] (10.3,3) -- (10.7,4.2);
\draw (11,4.2) circle (0.3);
\draw (11,4.2) circle (0.25);
\node at (11,4.2) {$q_{11}^2$};
\node[anchor=south,rotate=71.56] at (10.5,2.4) {$a$};
\draw[->, thick] (10.3,1.8) -- (10.7,3);
\draw (11,3) circle (0.3);
\draw (11,3) circle (0.25);
\node at (11,3) {$q_{11}^3$};
\node[anchor=south,rotate=71.56] at (10.5,1.2) {$a$};
\draw[->, thick, gray, dashed] (10.3,0.6) -- (10.7,1.8);
\draw[thick, gray, dashed] (11,1.8) circle (0.3);
\node at (11,1.8) {$q_{11}^4$};
\node[anchor=south] at (10.5,0.6) {$b$};
\draw[->, thick, gray, dashed] (10.3,0.6) -- (10.7,0.6);
\draw[thick, gray, dashed] (11,0.6) circle (0.3);
\node at (11,0.6) {$q_{11}^5$};
\node[anchor=south] at (10.5,-0.6) {$a$};
\draw[->, thick] (10.3,-0.6) -- (10.7,-0.6);
\draw (11,-0.6) circle (0.3);
\draw (11,-0.6) circle (0.25);
\node at (11,-0.6) {$q_{11}^6$};
\node[anchor=south] at (10.5,-1.8) {$a$};
\draw[->, thick, gray, dashed] (10.3,-1.8) -- (10.7,-1.8);
\draw[thick, gray, dashed] (11,-1.8) circle (0.3);
\node at (11,-1.8) {$q_{11}^{7}$};
\node[anchor=south,rotate=-71.56] at (10.5,-2.4) {$b$};
\draw[->, thick,gray, dashed] (10.3,-1.8) -- (10.7,-3);
\draw[thick, gray, dashed] (11,-3) circle (0.3);
\node at (11,-3) {$q_{11}^{8}$};
\node[anchor=south,rotate=-71.56] at (10.5,-3.6) {$a$};
\draw[->, thick, gray, dashed] (10.3,-3) -- (10.7,-4.2);
\draw[thick, gray, dashed] (11,-4.2) circle (0.3);
\node at (11,-4.2) {$q_{11}^{9}$};
\node[anchor=south,rotate=-71.56] at (10.5,-4.8) {$b$};
\draw[->, thick, gray, dashed] (10.3,-3) -- (10.7,-5.4);
\draw[thick, gray, dashed] (11,-5.4) circle (0.3);
\node at (11,-5.4) {$q_{11}^{10}$};
\node[anchor=south,rotate=-71.56] at (10.5,-6) {$b$};
\draw[->, thick, gray, dashed] (10.3,-4.2) -- (10.7,-6.6);
\draw[thick, gray, dashed] (11,-6.6) circle (0.3);
\node at (11,-6.6) {$q_{11}^{11}$};
\node[anchor=south,rotate=-71.56] at (10.5,-7.2) {$a$};
\draw[->, thick] (10.3,-5.4) -- (10.7,-7.8);
\draw (11,-7.8) circle (0.3);
\draw (11,-7.8) circle (0.25);
\node at (11,-7.8) {$q_{11}^{12}$};
\end{tikzpicture}
\end{center}
\end{figure}

\newpage

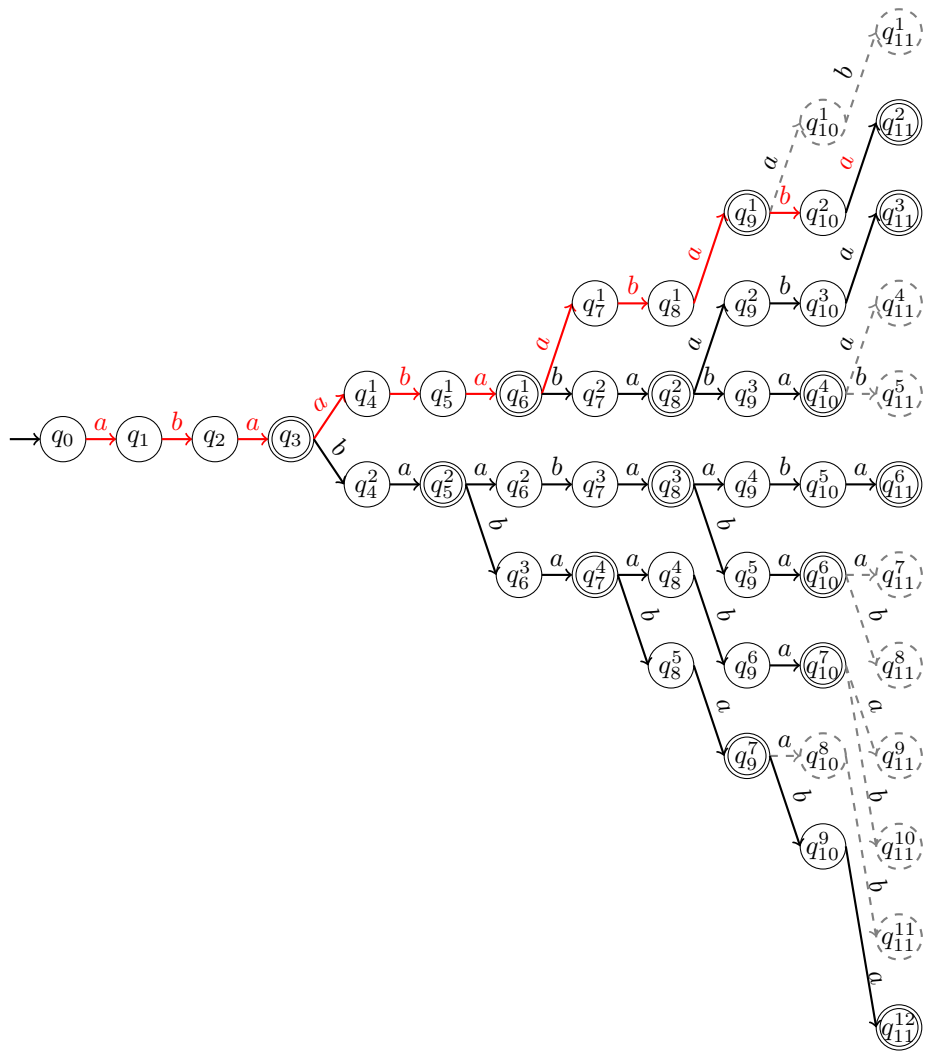
\end{document}